\documentclass[runningheads]{llncs}
\usepackage[linesnumbered,ruled,vlined]{algorithm2e}
\usepackage{algorithmic}
\usepackage{graphicx}
\usepackage{booktabs}
\usepackage{multirow}

\usepackage{ntheorem}
\usepackage{graphicx}  
\usepackage{float}  
\usepackage{subfigure}  
\usepackage{amssymb}
\usepackage{url}
\usepackage{appendix}
\usepackage{subcaption}
\usepackage{caption}

\begin{document}
\title{A Local Search Algorithm for MaxSMT(LIA)}


\author{Xiang He\inst{1,2}\orcidID{0009-0005-0730-1388}$^{\star}$ \and Bohan Li\inst{1,2}\orcidID{0000-0003-1356-6057}\thanks{These two authors are co-first authors, as they contribute equally.}
\and Mengyu Zhao\inst{1,2}\orcidID{0009-0001-8436-3532}  \and Shaowei Cai\inst{1,2}\orcidID{0000-0003-1730-6922}\thanks{Corresponding author}}



\institute{Key Laboratory of System Software (Chinese Academy of Sciences) and\\ State Key Laboratory of Computer Science,\\Institute of Software, Chinese Academy of Sciences, Beijing, China\\
\email{\{hexiang,libh,zhaomy,caisw\}@ios.ac.cn}
\and
School of Computer Science and Technology \\
University of Chinese Academy of Sciences, Beijing, China
}

\authorrunning{Xiang He, Bohan Li, Mengyu Zhao, and Shaowei Cai}

\maketitle              
\begin{abstract}
MaxSAT modulo theories (MaxSMT) is an important generalization of Satisfiability modulo theories (SMT) with various applications.
In this paper, we focus on MaxSMT with the background theory of Linear Integer Arithmetic, denoted as MaxSMT(LIA).
We design the first local search algorithm for MaxSMT(LIA) called PairLS, based on the following novel ideas.
A novel operator called {\it pairwise} operator is proposed for integer variables.
It extends the original local search operator by simultaneously operating on two variables, enriching the search space.
Moreover, a compensation-based picking heuristic is proposed to determine and distinguish the {\it pairwise} operations.
Experiments are conducted to evaluate our algorithm on massive benchmarks.
The results show that our solver is competitive with state-of-the-art MaxSMT solvers.
Furthermore, we also apply the {\it pairwise} operation to enhance the local search algorithm of SMT, which shows its extensibility.


\end{abstract}
\section{Introduction}

The maximum satisfiability problem (MaxSAT) is an optimization version of the satisfiability problem (SAT), aiming to minimize the number of falsified clauses, and it has various applications~\cite{li2021maxsat}.
A generalization of MaxSAT is the weighted Partial MaxSAT problem, where clauses are divided into hard and soft clauses with weights (positive numbers). 
The goal is to find an assignment that satisfies all hard clauses and minimizes the total weight of falsified soft clauses.
MaxSAT solvers have made substantial progress in recent years~\cite{ansotegui2013sat,martins2014open,lei2018solving,li2021combining,ignatiev2019rc2}.

However, MaxSAT has limited expressiveness, and it can be generalized from the Boolean case to Satisfiability Modulo Theories (SMT), deciding the satisfiability of a first-order logic formula with respect to certain background theories, leading to a  generalization called MaxSAT Modulo Theories (MaxSMT)~\cite{nieuwenhuis2006sat}. With its enhanced expressive power, MaxSMT has various practical applications, such as safety verification \cite{brockschmidt2015compositional}, concurrency debugging \cite{terra2019concurrency}, non-termination analysis \cite{larraz2014proving} and superoptimization \cite{albert2020synthesis}.
Compared to MaxSAT and SMT solving, the research on MaxSMT solving is still in its preliminary stage.
Cimatti et al.~\cite{cimatti2010satisfiability} introduced the concept of ``Theory of Costs'' and developed a method to manage SMT with Pseudo-Boolean (PB) constraints and minimize PB cost functions. 
Sebastiani et al.~\cite{sebastiani2012optimization,sebastiani2015optimization} proposed an approach to solve MaxSMT problem by encoding it into SMT with PB functions.
A modular approach for MaxSMT called Lemma-Lifting was proposed by Cimatti et al.~\cite{cimatti2013modular}, which involves the iterative exchange of information between a lazy SMT solver and a purely propositional MaxSAT solver. 
The implicit hitting set approach was lifted from the propositional level to SMT~\cite{fazekas2018implicit}.
Two well-known MaxSMT solvers are OptiMathSAT~\cite{sebastiani2020optimathsat} and $\nu Z$~\cite{bjorner2014nuz}, which are currently the state-of-the-art MaxSMT solvers.
In this paper, we focus on the MaxSMT problem with the background theory of Linear Integer Arithmetic (LIA), denoted as MaxSMT(LIA), which consists of arithmetic atomic formulas in the form of linear equalities or inequalities over integer variables.

We apply the local search method to solve MaxSMT(LIA).
Although local search has been successfully used to solve SAT~\cite{LiL12,BalintS12,cai2013local,cai2015ccanr,Biere17} and recently to SMT on the theory of bit-vector theory~\cite{frohlich2015stochastic,niemetz2016precise,niemetz2017propagation,niemetz2020ternary}, integer arithmetic~\cite{cai2022local,cai2023local} and real arithmetic~\cite{li2023local,li2023local2}, this is the first time that it is applied to MaxSMT.

First, we propose a novel operator for integer variables, named {\it pairwise operator}, to enrich the search space by simultaneously operating on two variables.
When the algorithm falls into the local optimum w.r.t. operations on a single variable, further exploring the neighborhood structure of {\it pairwise operator} can help it escape from the local optimum.


Moreover, a novel method based on the concept of {\it compensation} is proposed to determine the {\it pairwise operation}.
Specifically, the {\it pairwise operation} is determined as a pair of simultaneous modifications, one to satisfy a falsified clause, and the other to minimize the disruptions the first operation might wreak on the already satisfied clauses.
Then,  a two-level picking heuristic is proposed to distinguish these {\it pairwise operations}, by considering the potential of a literal becoming falsified.

Based on the above novel ideas, we design the first local search solver for MaxSMT(LIA) called PairLS, 
prioritizing hard clauses over soft clauses.
Experiments are conducted on massive benchmarks.
New instances based on SMT-LIB are generated to enrich the benchmarks for MaxSMT(LIA).
We compare our solver with 2 state-of-the-art MaxSMT(LIA) solvers, OptiMathSAT  and $\nu Z$.
Experimental results show that our solver is competitive with these state-of-the-art solvers.
We also present the evolution of solution quality over time, showing that PairLS can efficiently find a promising solution within a short cutoff time.
Ablation experiments are also conducted to confirm the effectiveness of proposed strategies.
Moreover, we  apply the {\it pairwise} operator to enhance the local search algorithm of SMT, demonstrating its extensibility.

\section{Preliminary}
\label{pre}

\subsection{MaxSMT on Linear Integer Arithmetics}
The Satisfiability modulo theories (SMT) problem determines the satisfiability of a given quantifier-free first-order formula with respect to certain background theories.
Here we consider the theory of Linear Integer Arithmetic (LIA), consisting of arithmetic formulae in the form of linear equalities or inequalities over integer variables ($\sum_{i=0}^n{a_ix_i\leq k}$ or $\sum_{i=0}^n{a_ix_i = k}$)\footnote{strict linear equalities in the form of ($\sum_{i=0}^n{a_ix_i< k}$) can be transformed to ($\sum_{i=0}^n{a_ix_i\leq k-1}$)}.
An atomic formula can be a propositional variable or an arithmetic formula.
A $literal$ is an atomic formula, or the negation of an atomic formula.
A $clause$ is the disjunction of a set of literals, and a formula in {\it conjunctive normal form (CNF)} is the conjunction of a set of clauses.
Given the sets of propositional variables and integer variables, denoted as $P$ and $X$ respectively, an assignment $\alpha$ is a mapping $X\rightarrow Z$ and $P\rightarrow \{false, true\}$, and $\alpha(x)$ denotes the value of a variable $x$ under $\alpha$. 

The (weighted partial) MaxSAT Modulo Theories problem (MaxSMT for short) is generated from SMT.
The clauses are divided into $hard$ clauses and $soft$ clauses with positive weight.

\begin{definition}
For a MaxSMT instance $F$, given the current assignment $\alpha$, if it satisfies all hard clauses, then $\alpha$ is a {\it feasible} solution, and the {\it cost} is defined as the total weight of all falsified soft clauses, denoted as $cost(\alpha)$.
\end{definition}

MaxSMT aims to find a feasible solution with minimal $cost$, that is, to find an assignment satisfying all hard clauses and minimizing the sum of the weights of the falsified soft clauses.
The MaxSMT problem with the background theory of LIA is denoted as MaxSMT(LIA).

\begin{example}
    Given a MaxSMT(LIA) formula $F=c_1\wedge c_2 \wedge c_3 \wedge c_4=(a-b\le 1\vee a-c\le0)\wedge(b-c\le -1)\wedge (a-d\le 1)\wedge(A)$, let $c_1$ and $c_4$ be hard clauses, $c_2$ and $c_3$ be soft clauses with weight 1 and 2.
    Given the current assignment $\alpha=\{a=0,b=0,c=0,d=0, A=true\}$,
    $cost(\alpha)=1$, since only $c_2$ is falsified.
\end{example}

\subsection{Local Search Components}
The {\it clause weighting scheme} is a popular local search method that associates an additional property (which is an integer number) called {\it penalty weight} to clauses and dynamically adjusts them to prevent the search from getting stuck in a local optimum.
We adopt the weighting scheme called Weighting-PMS~\cite{lei2018solving} to instruct the search. Weighting-PMS  has been applied in state-of-the-art local search solvers for MaxSAT, such as 
SATLIKE~\cite{lei2018solving} and
SATLIKE3.0~\cite{cai2020old}. 
When the algorithm falls into a local optimum, the Weighting-PMS dynamically adjusts the {\it penalty weights} of hard and soft clauses to guide the search direction.

Note that the {\it penalty weight} and the original weight of soft clauses are different. 
The goal of MaxSMT is to minimize the total original weight of unsatisfied soft clauses, while the {\it penalty weight} is updated during the search process, guiding the search in a promising direction.

Another key component of a local search algorithm is the {\it operator}, defining how to modify the current solution.
When an operator is instantiated by specifying the variable to operate and the value to assign, an {\it operation} is obtained.


\begin{definition}
The score of an operation $op$, denoted by $score(op)$, is the decrease of the total penalty weight of falsified clauses caused by applying $op$.
\end{definition}  

An operation is {\it decreasing} if its {\it score} is greater than 0.
Note that given a set of clauses, denoted as $C$, the $score$ of operation $op$ on the subformula composed of $C$ is denoted as $score_C(op)$.


\section{Review of LS-LIA}

\label{review}
As our algorithm adopts the two-mode framework of LS-LIA, which is the first local search algorithm for SMT(LIA)~\cite{cai2022local}, we briefly review it in this section. 


After the initialization, the algorithm switches between Integer mode and Boolean mode. In each mode, an operation on a variable of the corresponding data type is selected to modify the current assignment.
 The two modes switch to each other when the number of non-improving steps of the current mode reaches a threshold. 
 The threshold is set to $L\times P_b$ for the Boolean mode and $L\times P_i$ for the Integer mode, where $P_b$ and $P_i$ denote the proportion of Boolean and integer literals to all literals in falsified clauses, and $L$ is a parameter.


In the Boolean mode, the {\it flip} operator is adopted to modify a Boolean variable to the opposite of its current value. In the Integer mode as in Algorithm 1, a novel operator called {\bf critical move} ($cm$ for short) is proposed by considering the literal-level information.

\begin{definition}
The critical move operator, denoted as  $cm(x,\ell)$, assigns an integer variable $x$ to the threshold value making literal $\ell$ true, where $\ell$ is a falsified literal containing $x$.
\end{definition}

Specifically, the {\it threshold value} refers to the minimum modification to $x$ that can make $\ell$ true.
Example \ref{cm_example} is given to help readers understand the definition.
\begin{example}
\label{cm_example}
Given two falsified literals $\ell_1:(2a-b\le-3)$ and $\ell_2:(5c-d=5)$, and the current assignment is $\alpha=\{a=0,b=0,c=0,d=0\}$.
Then $cm(a,\ell_1)$, $cm(b,\ell_1)$, $cm(c,\ell_2)$, and $cm(d,\ell_2)$ refers to assigning $a$ to -2, assigning $b$ to $3$, assigning $c$ to 1 and assigning $d$ to $-5$ respectively.
\end{example}

An important property of the {\it critical move} operator is that after the execution, the corresponding literal must be true. Therefore, by picking a falsified literal and performing a $cm$ operation on it, we can make the literal true.

In our algorithm for MaxSMT(LIA), the {\it critical move} operator is also adopt-\\ed to make a falsified literal become true.

\begin{algorithm}[!t]
\caption{Integer Mode of LS-LIA}
\label{LS-Int}
\While{{\it non-improving steps} $\leq L\times P_i$}{
\lIf{all clauses are satisfied}{return }
\If{$\exists$  decreasing {\it cm} operation }{
$op:=$ such an operation with the greatest $score$}
\Else{
update penalty weights\;
$c:=$ a random falsified clause with integer variables\;
$op:=$ a {\it cm} operation in $c$ with  $score$\;
}
perform $op$ \; 
}

\end{algorithm}

\section{Pairwise Operator}
\label{pair}

In this section, we introduce a novel operator for integer variables,
denoted as {\it pairwise operator}.
It extends the original {\it critical move} operator to enrich the search space,
serving as an extended neighborhood structure.
We first introduce the motivation for the pairwise operator.
Then, based on pairwise operator, the framework of our algorithm in Integer mode is proposed.


\subsection{Motivation}

The original {\it critical move} operator only considers one single variable each time.
However, it may miss potential decreasing operations.
Specifically, when there exists no decreasing {\it critical move} operation,
operations that simultaneously modify two variables may be decreasing, which are not considered by {\it critical move}.

\begin{example}
\label{pair_example}
Given a formula $F=c_1\wedge c_2\wedge c_3=(a-b\le-2)\wedge(b-c\le1)\wedge(c-a\le1)$ where the penalty weight of each clause is 1.
and the current assignment is $\alpha=\{a=0,b=0,c=0\}$.
There exist two {\it critical move} operations: $cm(a,a-b\le-2)$ and $cm(b,a-b\le-2)$, referring to assigning $a$ to $-2$ and $b$ to $2$, respectively.
Both operations are not decreasing, since these two operations will respectively falsify $c_3$ and $c_2$.
However, simultaneously assigning $b$ to $2$ and $c$ to $1$ can be decreasing, since after the operation, all clauses become satisfied.
\end{example}

Thus, the {\it pairwise operator} simultaneously modifying two variables is proposed to find a decreasing operation when there is no decreasing $cm$ operation.
\begin{definition}
{\it Pairwise operator}, denoted as $p(v_1,v_2,val_1,val_2)$, will simultaneously modify $v_1$ to $val_1$ and $v_2$ to $val_2$ respectively, where $v_1$ and $v_2$ are integer variables, and $val_1$ and $val_2$ are integer parameters.
\end{definition}

The {\it pairwise operator} can be regarded as an extended neighborhood.
When there exists no decreasing {\it critical move} operation, indicating that the local optimum of modifying individual variables is found, the search space can be expanded by simultaneously modifying two variables, and the solution may be further improved, thanks to the following property:

\begin{proposition}
\label{pro_pair}
Given a {\it pairwise operation} $op_1=p(v_1,v_2,val_1,val_2)$, and two operations individually assigning $v_1$ to $val_1$ and $v_2$ to $val_2$, denoted as $op_2$ and $op_3$ respectively.
$op_1$ is decreasing while neither $op_2$ nor $op_3$ is decreasing,
only if there exists a clause $c$ containing both $v_1$ and $v_2$, and on clause $c$, $score_{\{c\}}(op_1)>score_{\{c\}}(op_2)+score_{\{c\}}(op_3)$.
\end{proposition}

The proof can be found in Appendix \ref{appedix}.
Recall the Example \ref{pair_example}, the {\it pairwise operation} that simultaneously assigns $b$ to 2 and $c$ to 1, denoted as $op_1$, is decreasing,
while none of the operations that individually assign $b$ to 2 and $c$ to 1, denoted as $op_2$ and $op_3$, is decreasing.
The reason lies in that $b$ and $c$ both appear in the clause $c_2$, and $score_{\{c_2\}}(op_1)>score_{\{c_2\}}(op_2)+score_{\{c_2\}}(op_3)$.

\section{Compensation-based Picking Heuristic}
\label{two-level}
To find a decreasing {\it pairwise operation} when there is no decreasing $cm$ operation,
we first introduce a method based on the concept of {\it compensation} to determine {\it pairwise operations}, which can satisfy the necessary condition in Proposition \ref{pro_pair} (Details can be found in Lemma \ref{necessary} of Appendix \ref{appendixb}).
Then, among these {\it pairwise operations}, we propose a two-level heuristic to distinguish them, by considering the potential of the compensated literals becoming falsified.

\subsection{Pairwise Operation Candidates for Compensation}
\label{detail implementation}
%

{\bf Motivation for compensation:} Since one variable may exist in multiple literals, changing a variable will affect all literals containing the variable, and may make some originally true literals become false.
 Moreover, if the literal is the reason for some clauses being satisfied, i.e., it is the only true literal in the clause, then falsifying the literal also falsifies the clause. 

  Formally, 
  for an operation $op$, we define a special set of literals $CL(op)=\{\ell | \ell$ is true and is the only true literal for some clauses, but $\ell$ would become false after individually performing $op\}$.
 After performing an operation $op$, the literals in the set $CL(op)$ are of special interest since some clauses containing such a literal would become falsified.

 {\bf Concept of compensation: }
 Let $op_1$ and $op_2$ denote two operations modifying individual variables. To minimize the disruptions that $op_1$ might wreak on the already satisfied clauses, another operation $op_2$ is simultaneously executed to make a literal $\ell\in CL(op_1)$ remain true under the assignment after operating $op_1$.
$op_2$ is denoted as {\it compensation} for $\ell$, and literals in the set $CL$ are denoted as {\it {\bf C}ompensated {\bf L}iterals}.



{\bf Compensation-based pairwise operation:}
A pairwise operation {\it p($v_1$,$v_2$, $val_1$,$val_2$)} can be regarded as simultaneously performing a pair of operations modifying individual variables, $op_1$ assigning $v_1$ to $val_1$ and $op_2$ assigning $v_2$ to $val_2$.
The procedure to determine $op_1$ and $op_2$ is described as follows.

First, a candidate $op_1$ is chosen to satisfy a falsified clause. To this end, we pick a variable $v_1$ from a false literal $\ell_1$ in a random falsified clause, and $op_1$ is the corresponding {\it cm} operation, $cm(v_1,\ell_1)$. It prioritizes literals from \textbf{hard} clauses and soft clauses are considered only when all hard clauses are satisfied.
To obtain sufficient candidates of $op_1$, $K$ (a parameter) literals are randomly selected from overall falsified clauses, and all variables in these literals are considered. 
The set of all candidate $op_1$ found in this stage is denoted as $CandOp$.

Second, given a literal $\ell_2\in CL(op_1)$, the $op_2$ w.r.t $op_1\in CandOp$ is determined to guarantee that $\ell_2$ remains true after simultaneously performing $op_1$ and $op_2$, meaning that $op_2$ is selected to {\it compensate} for $\ell_2$.
Specifically, to determine $op_2$, we pick a variable $v_2$ appearing in a literal $\ell_2\in CL(op_1)$, and calculate the value $val_2$  according to $cm(v_2,\ell_2)$ assuming $op_1$ performed.


\begin{example}
    Let us consider the formula presented in example \ref{pair_example}: $F=c_1\wedge c_2\wedge c_3=(a-b\le-2)\wedge(b-c\le1)\wedge(c-a\le1)$ where the penalty weight of each clause is 1,
    and the current assignment is $\alpha=\{a=0,b=0,c=0\}$.
    There is no decreasing {\it critical move} operation.
    As shown in Fig.~\ref{fig:pair}, performing $op_1=cm(b,a-b\le-2)$ that assigns $b$ to $2$ would falsify the literal $\ell=(b-c\le 1)$, the only true literal in $c_2$.
    To compensate for $\ell$, the operation $op_2$ that assigns $c$ to 1 is determined according to $cm(c,\ell)$, assuming that $op_1$ has been performed.
    All clauses become satisfied after simultaneously performing $op_1$ and $op_2$, and thus a decreasing pairwise operation $p(b,c,2,1)$ is found.

\end{example}

\begin{figure}
    \centering
    \includegraphics[width=0.5\linewidth]{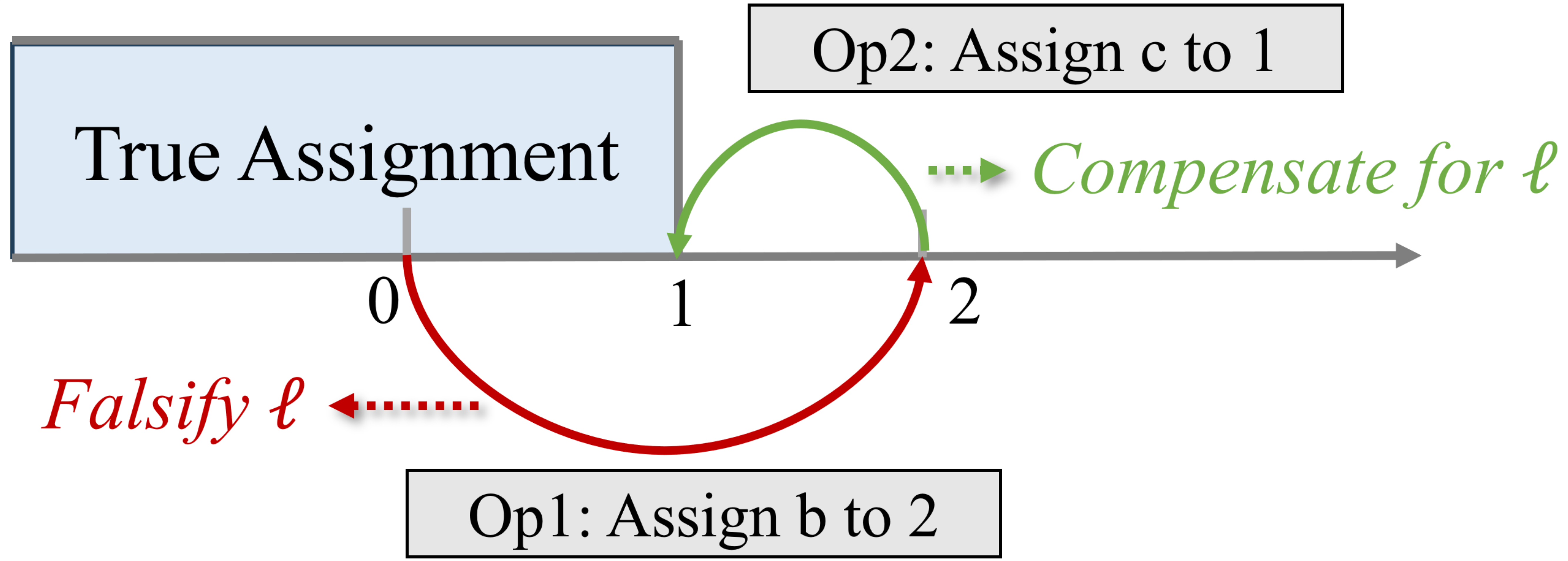}
    \caption{Given the literal $\ell=(b-c\le1)$, the axis refers to the value of $(b-c)$. Individually performing $op_1$ will falsify $\ell$, while $op_2$ can compensate for $\ell$.}
    \label{fig:pair}
\end{figure}

Note that there may exist multiple variables in the literal $\ell_2\in CL(op_1)$, and thus given the operation $op_1$, and the literal $\ell_2$ selected in the second step,  a set of pairwise operations is determined by considering all variables in $\ell_2$ except the variable in $op_1$, denoted as {\bf pair\_set($\ell_2$,$op_1$)}.

\subsection{Two-level Heuristic}

Among the literals selected  in the second step of determining a {\it pairwise operation},
we consider that some literals are more likely to become false, and should be given higher priority.
Thus, we distinguish such literals from others. 
They are formally defined as follows.

\begin{definition}
    Given an assignment $\alpha$, and a literal in the form of $\sum_{i=0}^n {a_ix_i  \leq k}$, we denote $\Delta=\sum_{i=0}^n{a_i\alpha(x_i)}-k$. The literal is a {\bf fragile} literal if $\Delta=0$ holds.
    Any true literal with $\Delta<0$ is {\bf safe}.
\end{definition}

A {\it fragile} literal with $\Delta=0$ is true as the inequality $\Delta\le0$ holds, but it can be falsified by any little disturbance that enlarges $\Delta$ of the corresponding {\it fragile} literal.
Comparatively, a literal is {\it safe} means that even if the value of a variable in the literal changes comparatively larger (as long as $\Delta\le0$ after the modification), it remains true.

\begin{example}
    Consider the formula: $F=l_1\wedge l_2\wedge l_3= (b-a\le -1) \wedge (a-c\le 0)\wedge (a-d\le3)$, where the current assignment is $\alpha=\{a=0,b=0,c=0,d=0\}$.
    $l_2$ and $l_3$ are two true literals.
    $l_2$ is a {\it fragile} literal since its $\Delta=0$, while $l_3$ is a {\it safe} literal since its $\Delta<0$.
    We consider that $l_2$ is more fragile than $l_3$, since a small disturbance, $cm(a,l_1)$ that assigns $a$ to 1, can falsify $l_2$ but not $l_3$.
\end{example}

In the second step of determining a pairwise operation,
among those {\it compensated literals},
we prefer  {\it fragile} literals and prioritize the corresponding pairwise operations.
Based on the intuition above, a two-level picking heuristic is defined:

\begin{itemize}
    \item We first choose the decreasing pairwise operation involving a {\it fragile} compensated literal.
    \item If there exists no such decreasing pairwise operation, we further select the pairwise operation involving {\it safe} compensated literals.
\end{itemize}


\subsection{Algorithm for Picking a Pairwise Operation}
Based on the picking heuristic, the algorithm for picking a pairwise operation is described in Algorithm \ref{pick pair}.
In the beginning, we initialize the set of pairwise operations involving {\it fragile} and {\it safe} compensated literals, denoted as $FragilePairs$ and $SafePairs$ (line 1).
Firstly, $K$ (a parameter) false literals are picked from overall falsified clauses, and all {\it critical move} operations in these literals are added into $CandOp$ (lines 2--7).
Note that it prioritizes hard clauses over soft clauses.


Then, for each operation $op_1\in CandOp$, we go through each compensated literal $\ell_2\in CL(op_1)$.
If $\ell_2$ is {\it fragile} (resp. {\it safe}), the set of corresponding {\it pairwise operations} determined by $pair\_set(\ell_2,op_1)$ are added to the $FragilePairs$ (resp. $SafePairs$) (line 8--13).

According to the two-level picking heuristic,
if there exist decreasing operations in $FragilePairs$,
we pick the one with the greatest $score$ (lines 14--15).
Otherwise, we pick a decreasing operation in $SafePairs$ if it exists (lines 16--17).
An operation with the greatest $score$ is selected via the BMS heuristic~\cite{cai2015balance}.
Specifically, the BMS heuristic samples $t$ pairwise operations (a parameter), and selects the decreasing one with the greatest $score$.

\begin{algorithm}[t]
\caption{pick\_pairwise\_op} 
\label{pick pair}
\KwOut{a decreasing pairwise operation if found}
   $FragilePairs:= \emptyset$,$SafePairs:= \emptyset$,
   $CandOp:= \emptyset$ $BestPair:= null$ \;
   \For {$i=1$ to $K$}{
        \If{$ {\exists}$  hard falsified clauses}{
		$\ell_1:=$ a random literal in a falsified hard clause \;
	}
 	\ElseIf{$ {\exists}$ soft falsified clauses}{
		$\ell_1:=$ a random literal in a falsified soft clause \;
	}
       $CandOp:= CandOp\bigcup\{cm(x,\ell_1)\vert x \; appears \; in \; \ell_1 \}$\;
  }
   \ForEach{$op_1$ in $CandOp$}{
       \ForEach{literal $\ell_2\in CL(op_1)$ }{
           \If{$\ell_2$ is fragile}{
              $FragilePairs:=FragilePairs$ $\cup$ $pair\_set(\ell_2,op_1)$\; 
          }
            \ElseIf {$\ell_2$ is safe}{
              $SafePairs:=SafePairs$ $\cup$ $pair\_set(\ell_2,op_1)$\;
           }
      }
  }
   \If{$ {\exists}$ decreasing operation in $FragilePairs$}{
    $BestPair:=$the operation with the greatest score picked by BMS\;
   }
   \ElseIf {$ {\exists}$ decreasing operation in $SafePairs$}{
      $BestPair:=$the operation with the greatest score picked by BMS\;
   }
   \Return $BestPair$\;
    
\end{algorithm}
\section{Local search algorithm}
\label{alg}

\begin{algorithm}[!t]		
\caption{Integer Mode of PairLS} 
\label{vnd_smt}
    \While{$non\_imp\_steps<MaxSteps $}{
	\If{$\not \exists$ falsified hard clauses AND  $cost(\alpha) < cost^*$}{
		$\alpha^* := \alpha,\; cost^* := cost(\alpha)$\;
	}
	\If{$ {\exists}$ decreasing critical move in hard falsified clauses}{
		$op := $ a decreasing critical move with the greatest score picked by BMS\;
	}
 	\ElseIf{$ {\exists}$ decreasing critical move in soft falsified clauses}{
		$op := $ a decreasing critical move with the greatest score picked by BMS\;
	}
        \lElse {$op:= pick\_pairwise\_op()$}
        \If{$op==null$}{
        update  penalty weights by Weighting-PMS\;
		\lIf{$\exists$ falsified hard clauses}{
				$c:=$ a random  falsified hard clause
		}
		\lElse{
				$c := $ a random falsified soft clause
		}
		$ op:= $ the critical move with the greatest $score$ in c\;
       }
            perform $op$ to modify $\alpha$\;
   }
\end{algorithm}

Based on the above novel components, we propose our algorithm for MaxSMT(L\\IA) called PairLS, prioritizing hard clauses over soft clauses.
PairLS initializes the complete current solution $\alpha$ by assigning all Integer variables to 0 and all Boolean variables to $false$.
Then, PairLS switches between Integer mode and Boolean mode.
When the time limit is reached, the best solution $\alpha^*$ and the corresponding best cost $cost^*$ are reported if a feasible solution can be found. Otherwise, ``No solution found'' is reported.


The  Integer mode of PairLS is described in Algorithm \ref{vnd_smt}.
The current solution $\alpha$ is iteratively modified until the number of non-improving steps $non\_impr\_step$ exceeds the threshold $bounds$ (line 1).
If a feasible solution with a smaller $cost$ is found, then the best solution $\alpha^*$ and the best cost $cost^*$ are updated (lines 2--3).
In each iteration, the algorithm first tries to find a decreasing {\it critical move} operation with the greatest score via BMS heuristic~\cite{cai2015balance} (line 4--7).
Note that it prefers to pick operations from falsified hard clauses, and falsified soft clauses are picked only if all hard clauses are satisfied.
If it fails to find any decreasing {\it critical move} operation, indicating that it falls into the local optimum of modifying individual variables, then it continues to search the neighborhood of {\it pairwise operation} (line 8).
If there exists no decreasing operation in both neighborhoods, the algorithm further escapes from the local optimum by updating the {\it penalty weight} (line 10), and satisfying a random clause by performing a {\it critical move} operation in it, preferring the one with the greatest {\it score} (lines 11--13).
Specifically, it also prioritizes hard clauses over soft clauses.

In the Boolean mode, the formula is reduced to a subformula that purely contains Boolean variables, which is indeed a MaxSAT instance.
Thus, our algorithm performs in the same way as SATLike3.0\footnote{ \url{https://lcs.ios.ac.cn/~caisw/Code/maxsat/SATLike3.0.zip}}, a state-of-the-art local search algorithm for MaxSAT~\cite{lei2018solving}.

\section{Experiments}
\label{experiments}
Experiments are conducted on 3 benchmarks to evaluate PairLS, comparing it with state-of-the-art MaxSMT solvers.
The promising experimental result indicates that our algorithm is efficient and effective in most instances.
We also present the evolution of solution quality over time, showing that PairLS can efficiently find promising solutions within a short time limit.
Moreover, the ablation experiment is carried out to confirm the effectiveness of our proposed strategies.

\subsection{Experiment Preliminaries}
{\bf Implementation:}
PairLS is implemented in C++ and compiled by g++ with the '-O3' option enabled.
There are 3 parameters in the solver:
$L$ for switching modes;
$t$ (the number of samples) for the BMS heuristic;
$K$ denotes the size of $CandOp$. 
The parameters are tuned according to our preliminary experiments and suggestions from the previous literature. They are set as follows:
$L=20$, $t=100$, $K=10$.


{\bf Competitors:}
We compare PairLS with 2 state-of-the-art MaxSMT solvers, namely OptiMathSAT(version 1.7.3) and $\nu Z$(version 4.11.2). 
OptiMathSAT applies MaxRes as the MaxSAT engine, denoted as Opt\_{res}, while the default configuration encodes the MaxSMT problem as an optimization problem, denoted as Opt\_{omt}.
$\nu Z$ also has 2 configurations based on the MaxSAT engines MaxRes and WMax, denoted as $\nu Z$\_{res}  and $\nu Z$\_{wmax}, respectively.
The binary code of OptiMathSAT and $\nu Z$ is downloaded from their websites.

{\bf Benchmarks:}
Our experiments are conducted on 3 benchmarks.
Those instances where the hard constraints are unsatisfiable are excluded, as they do not have feasible solutions.

    {\bf Benchmark MaxSMT-LIA:} This benchmark consists of 5520 instances generated based on SMT(LIA) instances from SMT-LIB\footnote{ \url {https://clc-gitlab.cs.uiowa.edu:2443/SMT-LIB-benchmarks/QF_LIA}}.
    The original SMT(LI\\A) benchmark consists of 690 instances from 3 families, namely bofill, convert, and wisa\footnote{{SMT(LIA) instances from other families are excluded because most of them are in the form of a conjunction of unit clauses, and thus the generation method is not applicable, since each produced soft assertion is also a hard assertion.}}.
    We adopt the same method to generate instances as in previous literature~\cite{fazekas2018implicit}: adding randomly chosen arithmetic atoms in the original problem with a certain proportion as unit soft assertions.
    4 proportions of soft clauses (denoted as $SR$) are applied, namely 10\%, 25\%, 50\% and 100\%.
     2 MaxSMT instances can be generated from each original SMT instance, based on different ways to associate soft clauses with  weights:
     one associates each soft clause with a unit weight of 1, and the other associates each soft clause with a random weight between 1 and the total number of atoms.
    Instances with unit weights and random weights are not distinguished as in~\cite{fazekas2018implicit}. 
    The total number of instances is $690\times2\times4=5520$, where 690 denotes the number of original SMT(LIA) instances, 2 denotes 2 kinds of weights associated with soft clauses, and 4 denotes the 4 proportions of soft clauses.
    Note that the ``bofill'' family was adopted in~\cite{fazekas2018implicit}, while the family of ``convert'' and ``wisa'' are new instances.
    
    
    {\bf Benchmark MaxSMT-IDL:} This benchmark contains 12888 new MaxSM-\\T instances generated by the above method, based on 1611 SMT(IDL) instances including all families from SMT-LIB\footnote{ \url{https://clc-gitlab.cs.uiowa.edu:2443/SMT-LIB-benchmarks/QF_IDL}} (similar to MaxSMT-LIA benchmark, the total number of instances is $1611\times2\times4=12888$).
    Instances with unit weights and random weights are also not distinguished when reporting results.
    
    {\bf Benchmark LL:} 
    The benchmark was proposed in~\cite{cimatti2013modular}.
    Unsatisfiable instances and instances over linear real arithmetic are excluded, resulting in 114 instances in total.
    56 instances contain soft clauses with unit weights of 1, and 58 instances contain soft clauses with random weights ranging from 1 to 100.
    Instances with {\it Unit} weights and {\it Random} weights are distinguished as in~\cite{cimatti2013modular}.
    



{\bf Experiment Setup:}
All experiments are conducted on a server with Intel Xeon Platinum 8153 2.00GHz and 2048G RAM under the system CentOS 7.7.1908.
Each solver executes one run for each instance in these benchmarks, as they contain sufficient instances.
The cutoff time is set to 300 seconds for the MaxSMT-LIA and MaxSMT-IDL benchmarks as previous work~\cite{fazekas2018implicit}, and 1200 seconds for the LL benchmark as previous work~\cite{cimatti2013modular}.

For each family of instances, we report the number of instances where the corresponding solver can find the best solution with the smallest $cost$ among all solvers, denoted by $\#win$, and the average running time to yield those best solutions, denoted as $time$.
Note that when multiple solvers find the best solution with the same $cost$ within the cutoff time, they are all considered to be winners. 
The solvers with the most $ \#win$ in the table are emphasized with {\bf bold} value.

The {\it solution} found by solvers and the corresponding time are defined as follows: As for complete solvers, $\nu Z$ and OptiMathSAT, we take the best upper bound found within the cutoff time as their {\it solution}, and the time to find such upper bound is recorded by referring to the log file.
Note that the proving time for complete solvers is excluded.
As for PairLS, the best {\it solution} found so far within the cutoff time and the time to find such a solution are recorded.


\subsection{Comparison to Other MaxSMT Solvers}
{\bf Results on benchmark MaxSMT-LIA:} As presented in Table \ref{lia_table}, PairLS shows competitive and complementary performance on this benchmark.
Regardless of the proportion of soft clauses $SR$, PairLS always leads in the {\it total} number of winning instances.
On the ``bofill'' family, PairLS performs better on instances with larger $SR$, confirming that PairLS is good at solving hard instances.
On the ``convert'' family, PairLS outperforms all competitors regardless of $SR$.
On the ``wisa'' family, PairLS cannot rival its competitors.
In Fig.3 of Appendix C,
we also present the run time comparison between PairLS and the best configuration of competitors, namely $\nu Z$\_{res} and Opt\_res. 
The run time comparison indicates that PairLS is more efficient than Opt\_res and is complementary to $\nu Z$\_{res}. 

{\bf Results on benchmark MaxSMT-IDL:} As presented in Table \ref{idl_table}, PairLS can significantly outperform all competitors regardless of the proportion of soft clauses.
In the overall benchmark, PairLS can find a better solution than all competitors on 53.5\% of total instances, and it can lead the best competitor by 1224 ``winning'' instances, confirming its dominating performance.
In Fig.4 of Appendix C,
we also present the run time comparison between PairLS and the best configuration of competitors, namely $\nu Z$\_{res} and Opt\_omt, indicating that PairLS is more efficient than competitors in instances with small $SR$. 

{\bf Results on LL benchmark:} 
The results are shown in Table \ref{LL_table}.
PairLS shows comparable but overall poor performance compared to its competitors on this benchmark.
One possible reason is that the front-end encoding for these benchmarks would generate many auxiliary Boolean variables, while PairLS cannot effectively explore the Boolean structure as LS-LIA~\cite{cai2022local}.
Specifically, the average number of auxiliary variables in this benchmark is 1220, while the counterparts in MaxSMT-LIA and MaxSMT-IDL are 327 and 528.

\begin{table}[]
\centering
\caption{Results on benchmark MaxSMT-LIA. The results are classified according to the proportion of soft clauses, $SR$. $Sum$ presents the overall performance.}
\label{lia_table}
\renewcommand\arraystretch{1.0}
\setlength{\tabcolsep}{0.5mm}\scalebox{1.0}{
\begin{tabular}{@{}llllllll@{}}
\toprule
\multirow{2}*{$SR$} & \multirow{2}*{family} & \multirow{2}*{\#inst} 
& \ \ Opt\_omt& \ \ Opt\_res&
$\ \ \ \nu Z$\_res& $\ \ \nu Z$\_wmax 
&\ \ \ PairLS \\
&&
&$\#win$($time$)
&$\#win$($time$)
&$\#win$($time$)
&$\#win$($time$)
&$\#win$($time$)\\ \midrule
\multirow{4}{*}{10\%}
&bofill & 814 &  773(17.1) & 777(7.1)& 758(7.9) & \textbf{797}(7.6) & 762(10.9)  \\
 &convert & 560 & 445(21.1) & 495(1.2) & 228(18.4)&8(103.1)  & \textbf{558}(1.1)   \\ 
  &wisa  & 6 & 4(27.4) & 3 (33.3)& \textbf{6} (11.3)& 0(0) & 0(0)      \\ [0.6ex] 
&    Total & 1380 & 1222(18.6)  & 1275(4.9) & 992(10.3) &805(8.5)& \textbf{1320}(6.7)  \\ \midrule

\multirow{4}{*}{25\%}
&bofill  &  814    & 736(59.2) & 720 (17.1) & 677(21.9) &\textbf{776}(21.8)& 641(35.4)  \\
 &convert & 560 & 415(16.3)   & 493(12.7)   & 205(13.3)& 5(125.6)  &\textbf{558}(1.4)    \\ 
  &wisa  & 6 &\textbf{4}(14.5) & 1(11.3)& \textbf{4}(73.2) & 0(0)  & 0(0)     \\ [0.6ex] 
&    Total & 1380 & 1155(43.6)  & \textbf{1214}(16.1) & 886(20.1) & 781(22.5) & 1199(19.6)    \\ \midrule

\multirow{4}{*}{50\%}
&bofill & 814 & 82(231.0) & 489(12.1) & 515(33.1) & 217(35.8) & \textbf{542}(66.7) \\
 &convert &  560 & 405(22.3)&   508(9.1)  & 177(19.8) & 0(0)  & \textbf{558}(1.2)   \\ 
  &wisa  & 6 & 3(36.6)& 1(11.3)& \textbf{5}(39.2)& 0(0)  & 0(0)      \\ [0.6ex] 
&    Total & 1380 & 490(57.3)  & 998(10.8) & 697(29.8)   & 217(39.7)  & \textbf{1100}(33.4)  \\ \midrule

\multirow{4}{*}{100\%}
&bofill & 814 & 0(0)& 0(0)& 402(75.6)& 15(20.7) & \textbf{601}(128.8) \\
 &convert &   560 & 399(27.3)  & 503(14.7)  & 185(19.3)& 19(55.0)   & \textbf{558}(1.2)     \\ 
  &wisa  &   6 & 1(162.2)& 1(193.3)& \textbf{6}(38.3) & 0(0)   & 0(0)     \\ [0.6ex] 
&    Total & 1380 & 400(27.7) & 504(15.0) & 593(57.6)  & 34(39.9) &\textbf{1159}(67.3)   \\ \midrule

$Sum$& & 5520 & 3267(34.3) & 3937(10.9)  &3168(26.1)  &1837(18.7)  & \textbf{4778}(30.7) \\ \bottomrule

\end{tabular}
}
\end{table}






\begin{table}[]
\centering
\caption{Results on MaxSMT-IDL benchmark. The results are classified according to the proportion of soft clauses, $SR$. $Sum$ presents the overall performance.}
\label{idl_table}
\renewcommand\arraystretch{1.0}
\setlength{\tabcolsep}{1.2mm}\scalebox{1.0}{
\begin{tabular}{@{}lllllll@{}}
\toprule
\multirow{2}*{$SR$} & \multirow{2}*{\#inst} 
&\ \  Opt\_omt&\ \  Opt\_res&
$\ \ \ \nu Z$\_res& \ \ $\nu Z$\_wmax 
&\ \ \ PairLS \\
&
&$\#win$($time$)
&$\#win$($time$)
&$\#win$($time$)
&$\#win$($time$)
&$\#win$($time$) \\ \midrule
10\%& 3222 & 1086(248.1) & 873(201.6) & 1151(238.4) & 839(221.5) & \textbf{1542}(159.9) \\ 
25\% & 3222 & 1171(195.2) & 934(178.6) & 1502(192.7) & 722(292.1) & \textbf{1744}(148.0) \\ 
50\%& 3222 & 1094(195.8) & 890(183.2) & 1461(198.2) & 758(290.9) & \textbf{1829}(149.3) \\ 
100\% &  3222 & 1061(204.1)& 880(198.3) & 1559(201.6) & 934(260.4) & \textbf{1782}(157.1)  \\ \midrule
$Sum$ & 12888 & 4412(210.5)  &  3577(190.2)  &   5673(205.8)&   3253(264.5) & \textbf{6897}(153.3) \\ \bottomrule

\end{tabular}
}
\end{table}

\begin{table}[]
\renewcommand\arraystretch{1.0}
\centering
\caption{Results on LL benchmark. Instances with {\it Unit} weights and {\it Random} weights are distinguished. $Sum$ presents the overall performance. }
\label{LL_table}
\setlength{\tabcolsep}{1.0mm}\scalebox{1.0}{
\begin{tabular}{@{}ccccccc@{}}
\toprule
 \multirow{2}*{Category} & \multirow{2}*{\#inst} 
&Opt\_omt& Opt\_res&
$\nu Z$\_res&$\nu Z$\_wmax 
& PairLS \\
&
&$\#win$($time$)
&$\#win$($time$)
&$\#win$($time$)
&$\#win$($time$)
&$\#win$($time$)\\ \midrule

{\it Unit}&  56 & 41(117.3) & 32(130.3) & \textbf{49}(183.2)& 15(40.0) & 23(0.4)  \\ 
{\it Random}&  58 & 45(117.9) & 37(120.3) &\textbf{53}(100.1)  & 17(39.0) & 22(0.4) \\ 
$Sum$&  114 & 86(117.6)& 69(124.9) & \textbf{102}(140.0)& 32(39.5) & 45(0.4)  \\ \midrule

\end{tabular}
}
\end{table}

\subsection{Evolution of Solution Quality}
To be more informative in understanding how the solvers compare in practice, the evolution of the solution quality over time is presented.
Specifically, we evaluate the overall performance on the MaxSMT-LIA and MaxSMT-IDL benchmark with 4 cutoff times, denoted as {\it cutoff}: 50, 100, 200, 300 seconds.
Given an instance, the proportion of the $cost$ to the sum of soft clause weights is denoted as $cost_P$~\footnote{If no feasible solution is found, then $cost_P$ is set as 1. Note that we present the average $cost_P$ rather than the average $cost$, since the $cost$ of certain instances can be quite large, dominating the average $cost$.}.
The average $cost_P$ over time is presented in Fig.\ref{ev},
showing that PairLS can efficiently find high-quality solutions within a short time.
Moreover, we also report the ``winning'' instances over time.
As shown in Table 7 in Appendix D,  
on each benchmark, PairLS leads the best competitor by at least 645 ``winning'' instances regardless of the cutoff time, confirming its dominating performance.
\subsection{Effectiveness of Proposed Strategies}
To analyze the effectiveness of our proposed strategies, two modified versions of PairLS are proposed as follows.

\begin{itemize}
    \item To analyze the effectiveness of {\it pairwise operation}, we modify PairLS by only using the {\it critical move} operator, leading to the version $v_{no\_pair}$.
    \item To analyze the effectiveness of {two-level heuristic} in {\it compensation-based picking heuristic} for picking a pairwise operation, PairLS is modified by selecting pairwise operation without distinguishing the {\it fragile} and {\it safe} compensated literals, leading to the version $v_{one\_level}$.

\end{itemize}
We compare PairLS with these modified versions on 3 benchmarks.
The results of this ablation experiment are presented in Table \ref{com_table}, confirming the effectiveness of the proposed strategies.

Moreover, we also analyze the extension for simultaneously operating on more variables. 
PairLS is modified  by simultaneously modifying three variables, where the third variable is modified to compensate for the second one, leading to the version  $v_{tuple}$. 
We conduct our experiments on MaxSMT-LIA. The results are in Apendix E.
When $N = 3$, the number of possible operations increases from O($k^2$) to O($k^3$), where $k$ is the number of variables in unsatisfied clauses. This might significantly slow down the searching process, indicating that modifying 2 variables simultaneously is the best choice of trade-off between cost and effectiveness.
\begin{figure}[H]
\centering  
\subfigure[MaxSMT-LIA]{
    \includegraphics[width=0.48\linewidth]{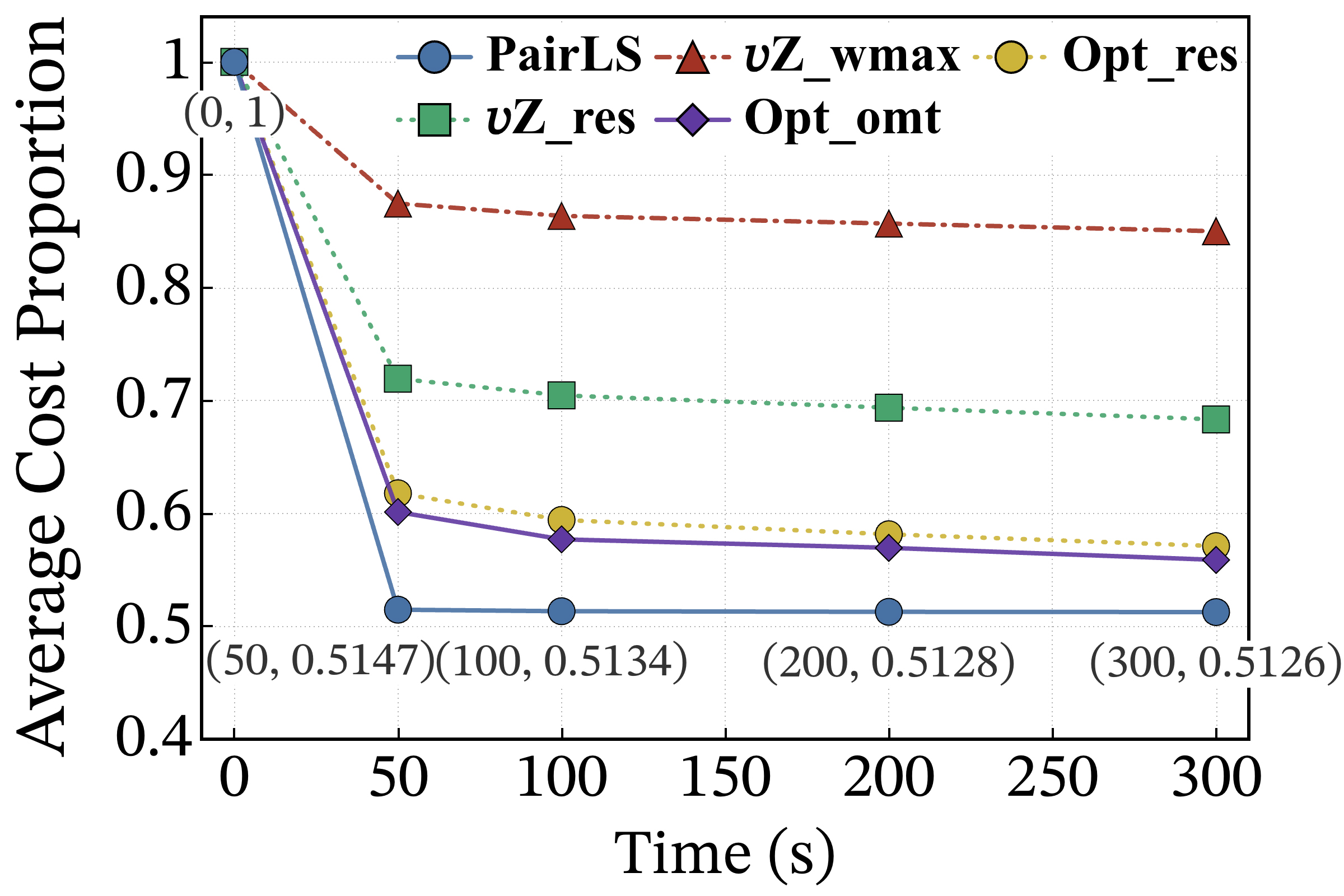}}\vspace{-1.5mm}
\subfigure[MaxSMT-IDL]{
    \includegraphics[width=0.48\linewidth]{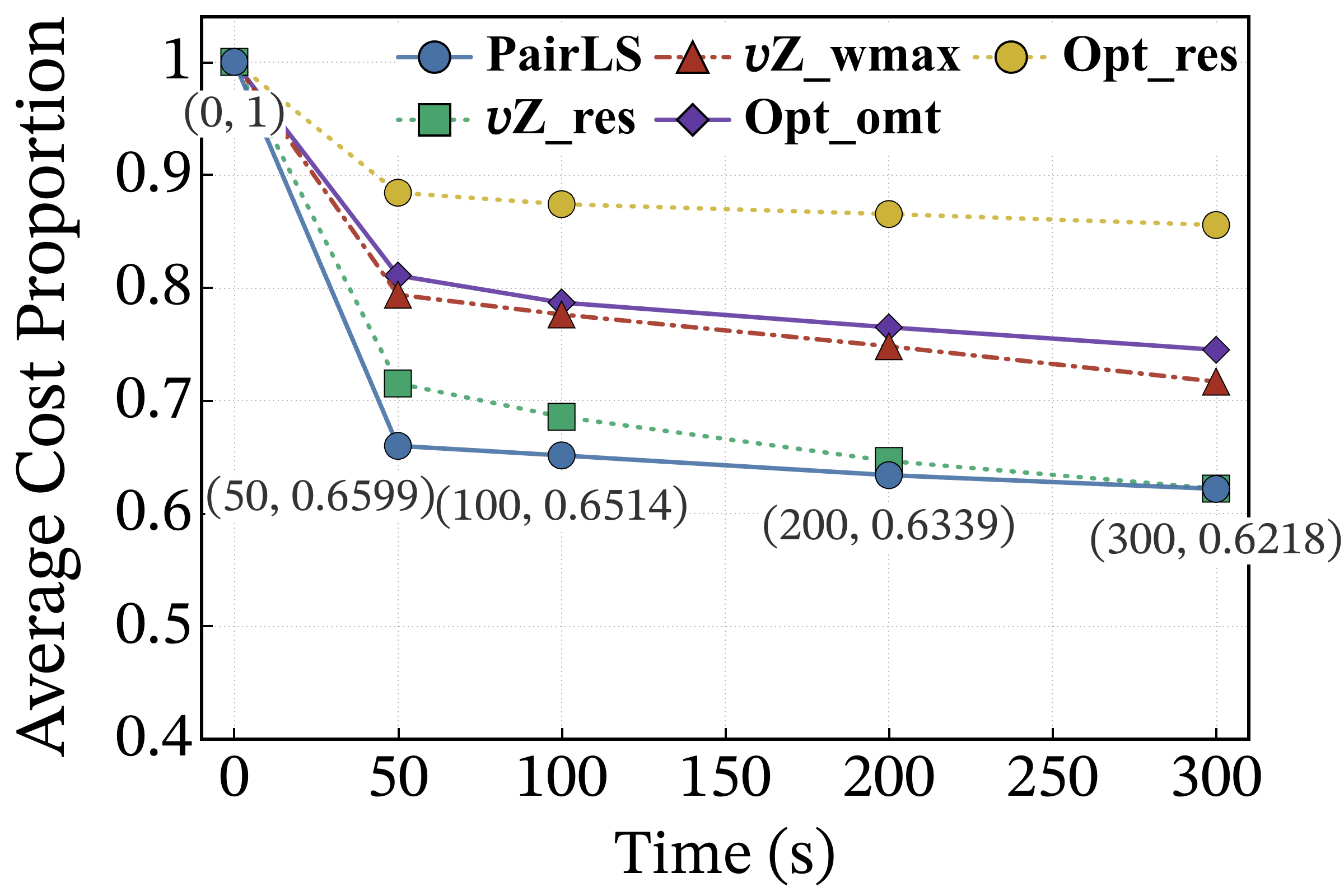}}\vspace{-1.5mm}
\caption{Evolution of average $cost_P$}
\label{ev}
\end{figure}


\begin{table}[]
    \centering
    \caption{Comparing PairLS with its modified versions. The number of instances
where PairLS performs better and worse are presented, denoted as \#better and \#worse respectively.
An algorithm is better than its competitor on a certain instance if it can find a solution with a lower $cost$. 
}
    \label{com_table}
    \renewcommand\arraystretch{1.0}
\setlength{\tabcolsep}{4mm}\scalebox{1.0}{
        \begin{tabular}{@{}llllll@{}}
            \toprule
            \multirow{2}{*}{} & \multirow{2}{*}{\#inst} & \multicolumn{2}{c}{$v_{no\_pair}$} & \multicolumn{2}{c}{$v_{one\_level}$} \\ \cmidrule(lr){3-4} \cmidrule(lr){5-6}
            & & $\#better$ & $\#worse$ & $\#better$ & $\#worse$ \\ \midrule
            MaxSMT-LIA & 5520 & {\bf 1834} & 65 & {\bf 705} &457 \\ 
            MaxSMT-IDL & 12888 & {\bf 3242} & 1962 & {\bf 3005} & 1826 \\ 
            LL & 114 & {\bf 27} & 0& {\bf 5} & 0\\ \bottomrule
        \end{tabular}
    }
\end{table}



\section{Discussion on the Extension of Pairwise Operation}
Since {\it pairwise} operator can be adapted to SMT(LIA) instances without additional modifications,
a potential extension is incorporating it into the local search algorithm for SMT(LIA).
When there is no decreasing  {\it cm} operation in the integer mode of LS-LIA (Algorithm 1 in Page 5), we apply {\it pairwise} operator to LS-LIA to enrich the search space as in PairLS, resulting in the corresponding version called LS-LIA-Pair. 
We compare LS-LIA-Pair with LS-LIA and other complete SMT solvers on SMT(LIA) instances, reporting the number of \textbf{unsolved instances} for each solver (Details are in Appendix F).
Without any \textbf{specific customization} tailored for SMT,
in both categories,  LS-LIA-Pair can solve 20 more instances compared to LS-LIA, which demonstrates that {\it pairwise} operator is an extensible method and could be further explored to enhance  the local search algorithm for SMT.



\section{Conclusion and Future Work}
In this paper, we propose the first local search algorithm for MaxSMT(LIA), called PairLS, based on the following components.
A novel {\it pairwise} operator is proposed to enrich the search space.
A compensation-based picking heuristic is proposed to determine and distinguish {\it pairwise} operations.
Experiments show that PairLS is competitive with state-of-the-art MaxSMT solvers, and {\it pairwise} operator is a general method.
Moreover, we also would like to develop a local search algorithm for MaxSMT on non-linear integer arithmetic and Optimization Modulo Theory problems.
Lastly, we hope to combine PairLS with complete solvers, since PairLS can efficiently find a solution with promising $cost$, serving as an upper bound for complete solvers to prune the search space.

\section*{Acknowledgements}{This research is supported by NSFC Grant 62122078 and CAS Project for Young Scientists in Basic Research (Grant No.YSBR-040).}

\label{conclusion}

\bibliographystyle{splncs04}


\newpage
\appendix
\section{Proof of Proposition 1}
\label{appedix}
\begin{proof}

The sets of clauses containing variables $v_1$ and $v_2$ are denoted as $C_{v_1}$ and $C_{v_2}$, respectively.
For simplicity, we denote the sets of clauses where only either variable $v_1$ or $v_2$ occur (but not both) as $C_{\bar{v_1}}=C_{v_1}\setminus(C_{v_1}\wedge C_{v_2})$ and $C_{\bar{v_2}}=C_{v_2}\setminus(C_{v_1}\wedge C_{v_2})$, respectively.
If the pairwise operation $op_1$ is decreasing, then the following inequality holds:
\begin{equation}
score(op_1)=score_{C_{\bar{v_1}}}(op_1)+score_{C_{\bar{v_2}}}(op_1)+score_{C_{v_1}\wedge C_{v_2}}(op_1)>0 \label{d_pair1}
\end{equation}

While individually operating $op_2$ and $op_3$ are not decreasing:
\begin{equation}
score(op_2)=score_{C_{\bar{v_1}}}(op_2)+score_{C_{v_1}\wedge C_{v_2}}(op_2)\le0 \label{d_pair2}
\end{equation}
\begin{equation}
score(op_3)=score_{C_{\bar{v_2}}}(op_3)+score_{C_{v_1}\wedge C_{v_2}}(op_3)\le0 \label{d_pair3}
\end{equation}

Note that, since $C_{\bar{v_1}}$ and $C_{\bar{v_2}}$ only either consist of $v_1$ or $v_2$ (but not both), $score_{C_{\bar{v_1}}}(op_1)=score_{C_{\bar{v_1}}}(op_2)$ and $score_{C_{\bar{v_2}}}(op_1)=score_{C_{\bar{v_2}}}(op_3)$.
By (\ref{d_pair1})-(\ref{d_pair2})-(\ref{d_pair3}), it can be inferred from the above inequality that:
\begin{equation}
score_{C_{v_1}\wedge C_{v_2}}(op_1)>score_{C_{v_1}\wedge C_{v_2}}(op_2)+score_{C_{v_1}\wedge C_{v_2}}(op_3) \label{d_pair4}
\end{equation}

{\bf Assumption:} assume that there does not exist any clause $c$ containing both $v1$ and $v2$ ($c\in C_{v_1}\wedge C_{v_2}$), and $score_{\{c\}}(op_1)>score_{\{c\}}(op_2)+score_{\{c\}}(op_3)$, the following inequality holds:
\begin{equation}
\forall c\in C_{v_1}\wedge C_{v_2}, score_{\{c\}}(op_1)\le score_{\{c\}}(op_2)+score_{\{c\}}(op_3) \label{d_pair5}
\end{equation}

According to (\ref{d_pair5}), by summing up the $score_{\{c\}}$ on each clause $c\in C_{v_1}\wedge C_{v_2}$ we can conclude that:

\begin{equation}
score_{C_{v_1}\wedge C_{v_2}}(op_1)\le score_{C_{v_1}\wedge C_{v_2}}(op_2)+score_{C_{v_1}\wedge C_{v_2}}(op_3) \label{d_pair6}
\end{equation}

(\ref{d_pair4}) and (\ref{d_pair6}) contradicts, and thus our assumption is false.
So we can conclude that there should be at least one clause $c$ containing both $v_1$ and $v_2$, and on $c$, $score_{\{c\}}(op_1)>score_{\{c\}}(op_2)+score_{\{c\}}(op_3)$.
\hfill$\square$
\end{proof}

\section{Compenstation-based pairwise operation can satisfy necessary condition of Proposition 1 }
\label{appendixb}
 \begin{lemma}
\label{necessary}
The compensation-based pairwise operation can satisfy the necessary condition proposed in Proposition \ref{pro_pair} to find a decreasing pairwise operation when there is no decreasing $cm$ operation. 
\end{lemma}
\begin{proof}
    A compensation-based pairwise $op_1=p(v_1,v_2,val_1,val_2)$ is regarded as simultaneously performing a pair of operations modifying individual variables, $op_2$ assigning $v_1$ to $val_1$ and $op_3$ assigning $v_2$ to $val_2$.
    According to the concept of compensation, there exists a true literal $\ell\in CL(op_2)$ which is the only true literal for some clauses, but $\ell$ would become false after individually performing $op_2$.
    Moreover, $\ell$ remains true after performing $op_1$.
    
    Let $c$ denote the satisfied clause containing such compensated literal $\ell\in CL(op_2)$ as the only true literal in it.
    Since $\ell$ contains both $v_1$ and $v_2$, $c$ contains both $v_1$ and $v_2$.

    According to the definition of $score$, given a satisfied clause $c$, if it is falsified by an operation $op$, then $score_{\{c\}}(op)<0$. If $c$ remains satisfied after performing $op$, then $score_c(op)=0$:

    Since $\ell$ is the only true literal in $c$ and $\ell$ can be falsified by performing $op_2$, then $c$ will be falsified by $op_2$:
    \begin{equation}
    score_{\{c\}}(op_2)<0
    \label{d_pair7}
\end{equation}

Since $c$ has already been a satisfied clause, it can only be falsified or remain true after performing $op_3$:
\begin{equation}
    score_{\{c\}}(op_3)\le0
    \label{d_pair8}
\end{equation}

Adding (\ref{d_pair7}) and (\ref{d_pair8}), we can conclude:
\begin{equation}
    score_{\{c\}}(op_2)+score_{\{c\}}(op_3)<0
    \label{d_pair9}
\end{equation}

Since $\ell$ remains true after performing $op_1$, then $c$ remains satisfied after performing $op_1$.
\begin{equation}
    score_{\{c\}}(op_1)=0
    \label{d_pair10}
\end{equation}

According to (\ref{d_pair9}) and (\ref{d_pair10}), we can conclude that $c$ is a common clause that contains both $v_1$ and $v_2$, and  $score_{\{c\}}(op_1)>score_{\{c\}}(op_2)+score_{\{c\}}(op_3)$.
Thus the compensation-based pairwise operation satisfies the necessary condition proposed in Proposition \ref{pro_pair} to find a decreasing pairwise operation when there is no decreasing $cm$ operation.
\hfill$\square$
\end{proof}

\section{Experiment of Time}
In the MaxSMT-LIA and MaxSMT-IDL benchmark, scatter plots are adopted to present the runtime comparison between PairLS and competitors with the best configuration, if both solvers can find a solution with the same $cost$.
The x-axis (resp. y-axis) denotes the runtime of PairLS (resp. competitor). Thus, nodes above the diagonal indicate instances that can be solved faster by PairLS.

\begin{figure}
    \centering

    \begin{minipage}{0.24\linewidth}
        \subfigure[$SR$=10\%]{
            \includegraphics[width=\textwidth]{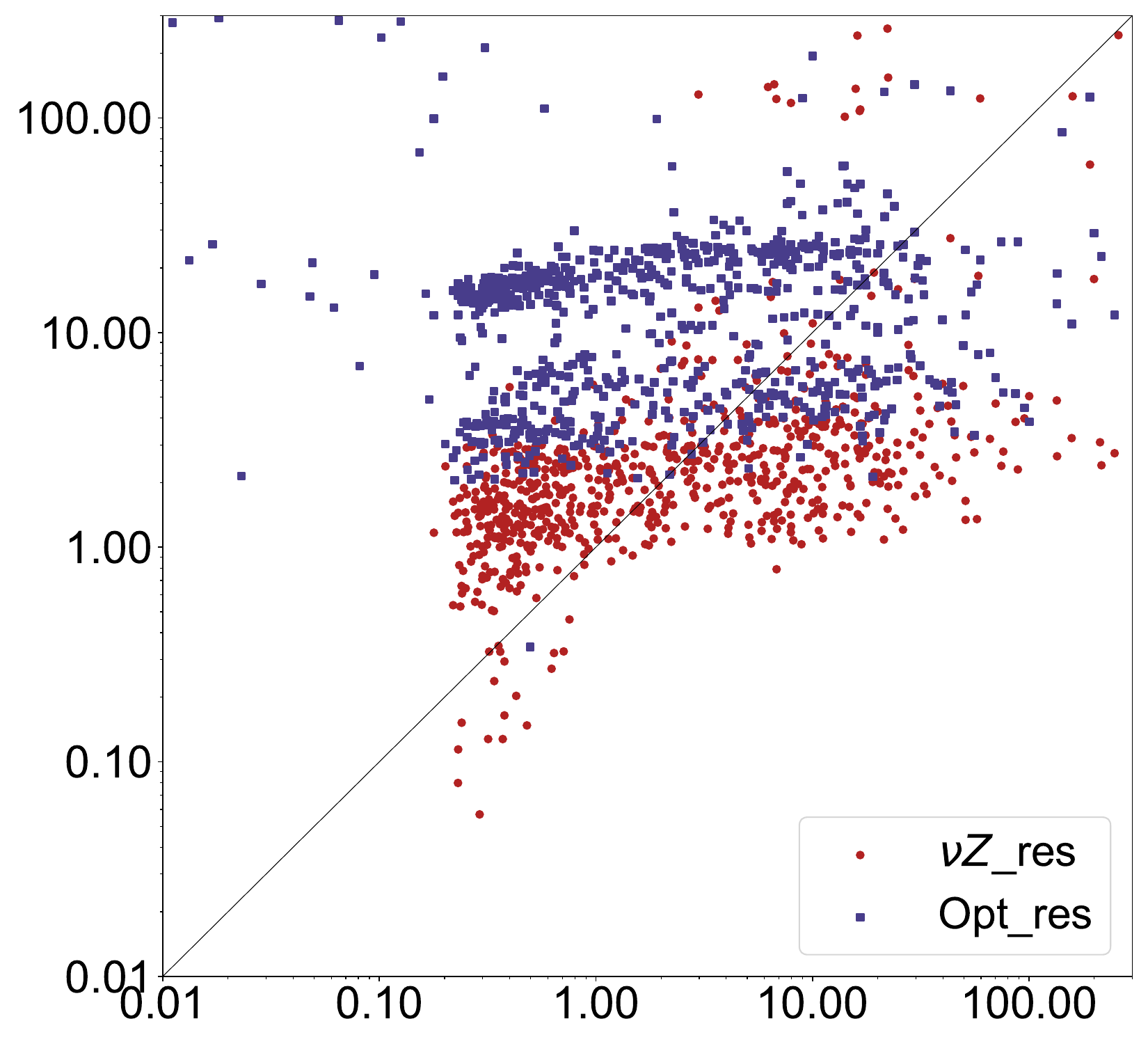}
        }
    \end{minipage}
    \begin{minipage}{0.24\linewidth}
        \subfigure[$SR$=25\%]{
            \includegraphics[width=\textwidth]{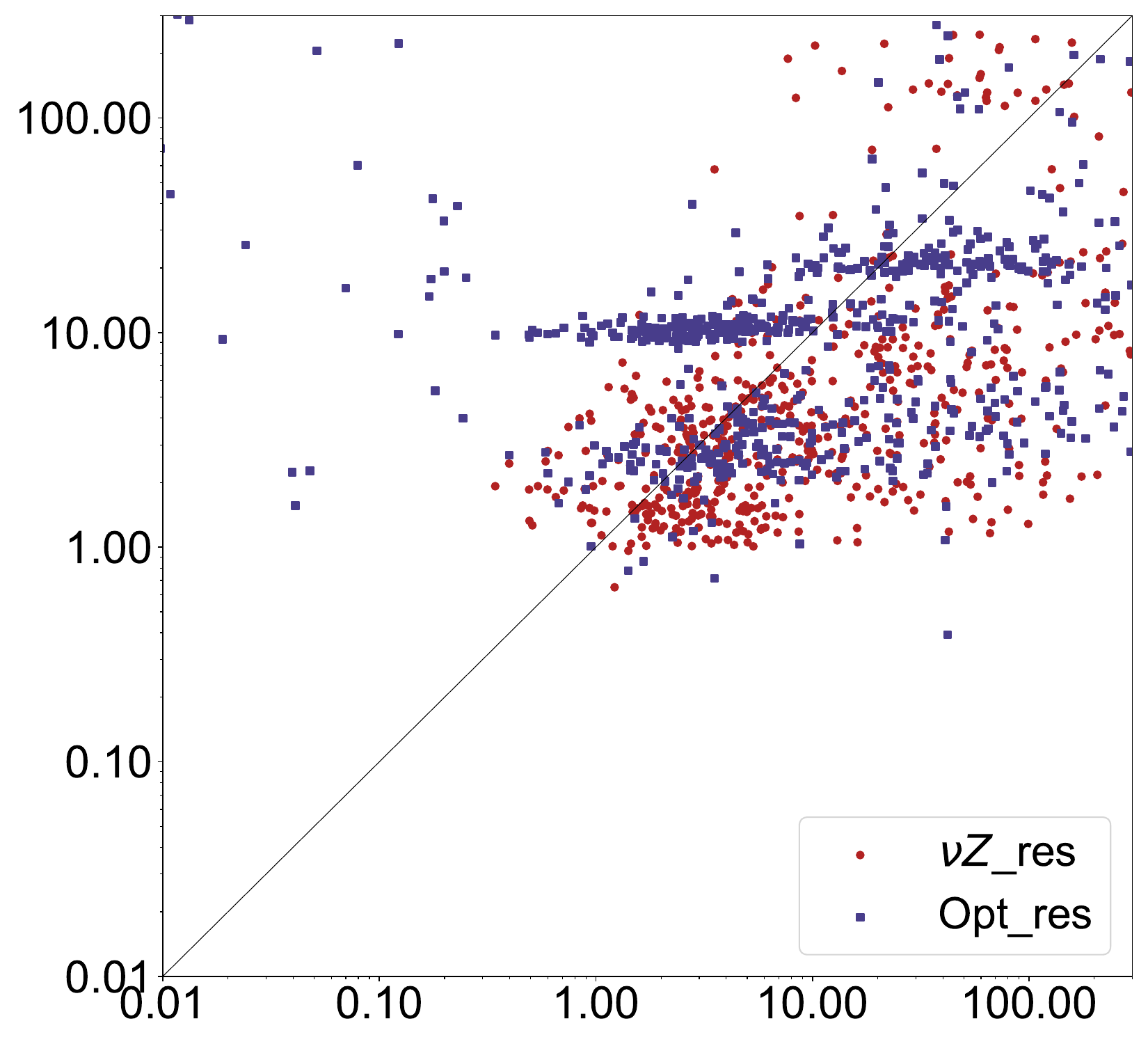}
        }
    \end{minipage}
    \begin{minipage}{0.24\linewidth}
        \subfigure[$SR$=50\%]{
            \includegraphics[width=\textwidth]{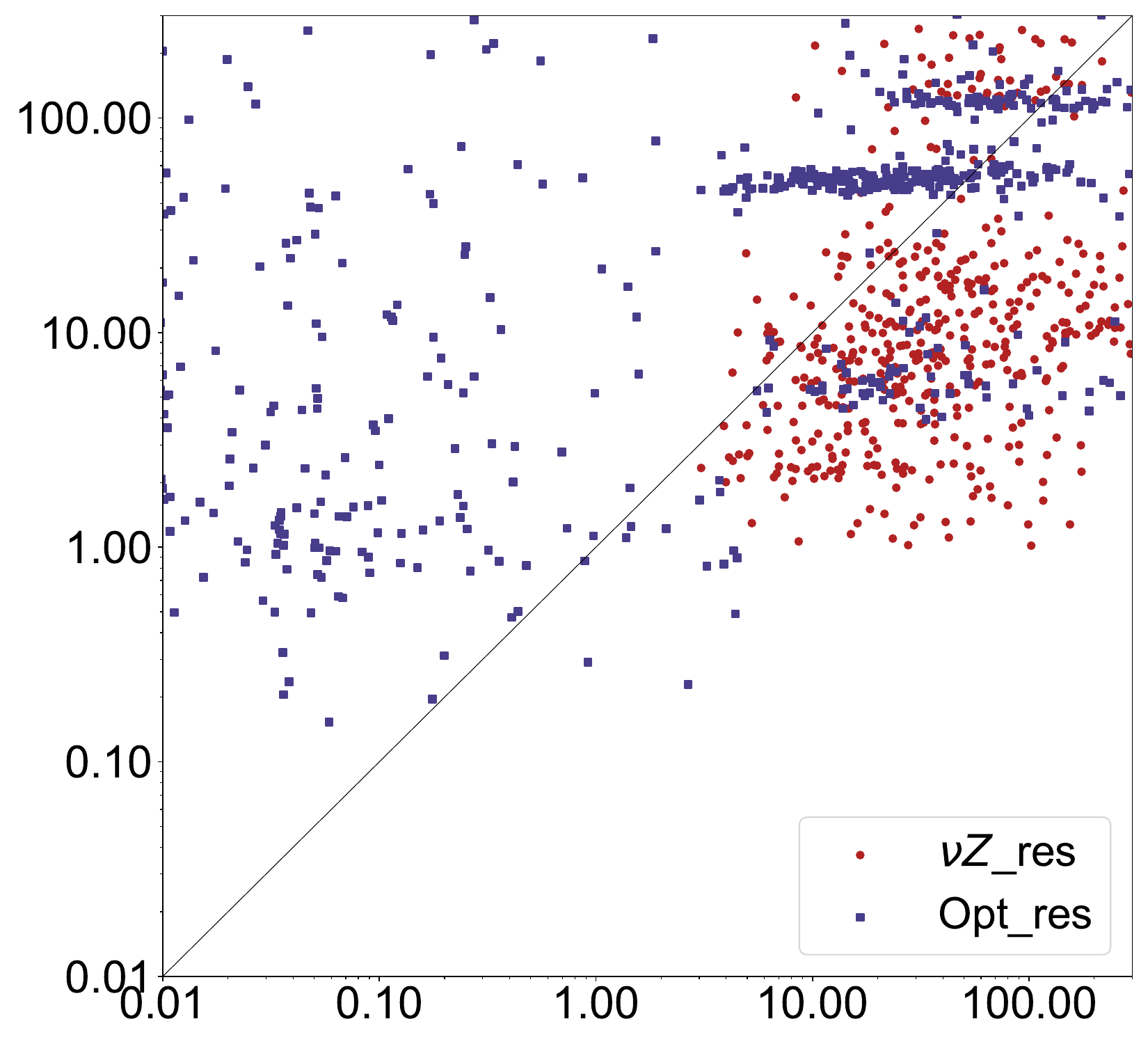}
        }
    \end{minipage}
    \begin{minipage}{0.24\linewidth}
        \subfigure[$SR$=100\%]{
            \includegraphics[width=\textwidth]{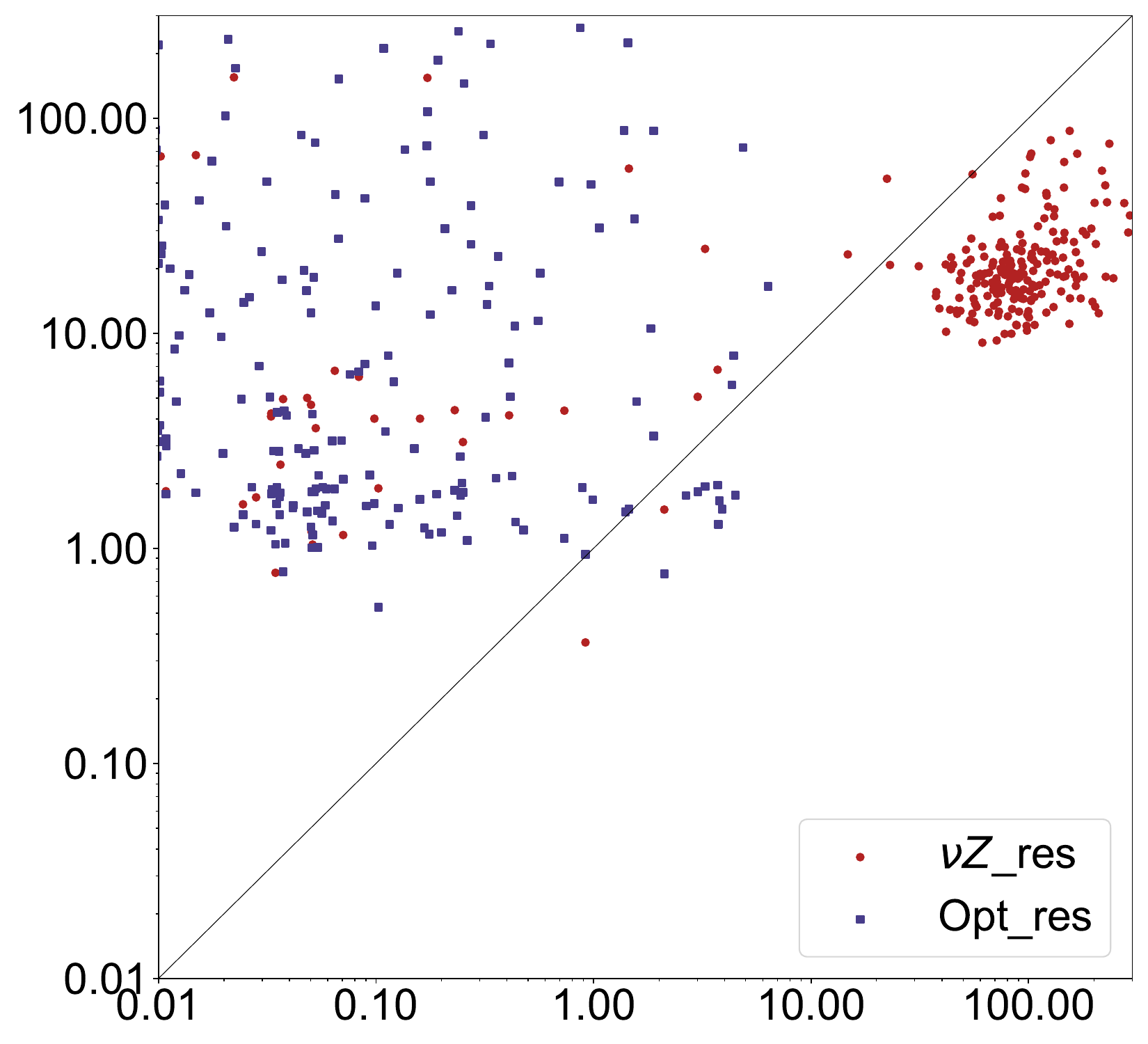}
        }
    \end{minipage}

    \caption{Run time comparison on MaxSMT-LIA. }
    \label{fig_compare_lia}
\end{figure}

\begin{figure}
    \centering

    \begin{minipage}{0.24\linewidth}
        \subfigure[$SR=10\%$]{
            \includegraphics[width=\textwidth]{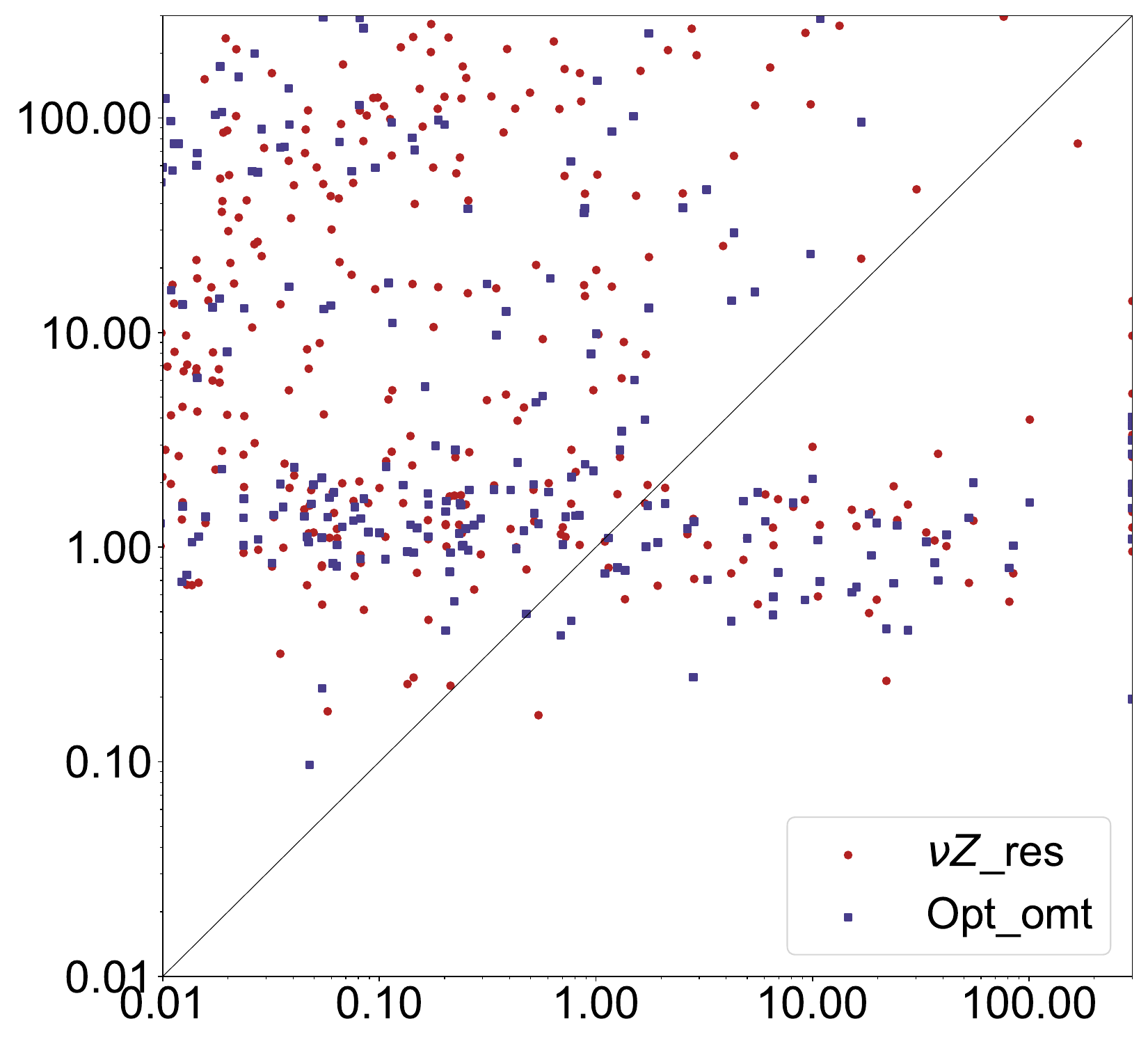}\vspace{-1.5mm}
        }
    \end{minipage}
    \begin{minipage}{0.24\linewidth}
        \subfigure[$SR=25\%$]{
            \includegraphics[width=\textwidth]{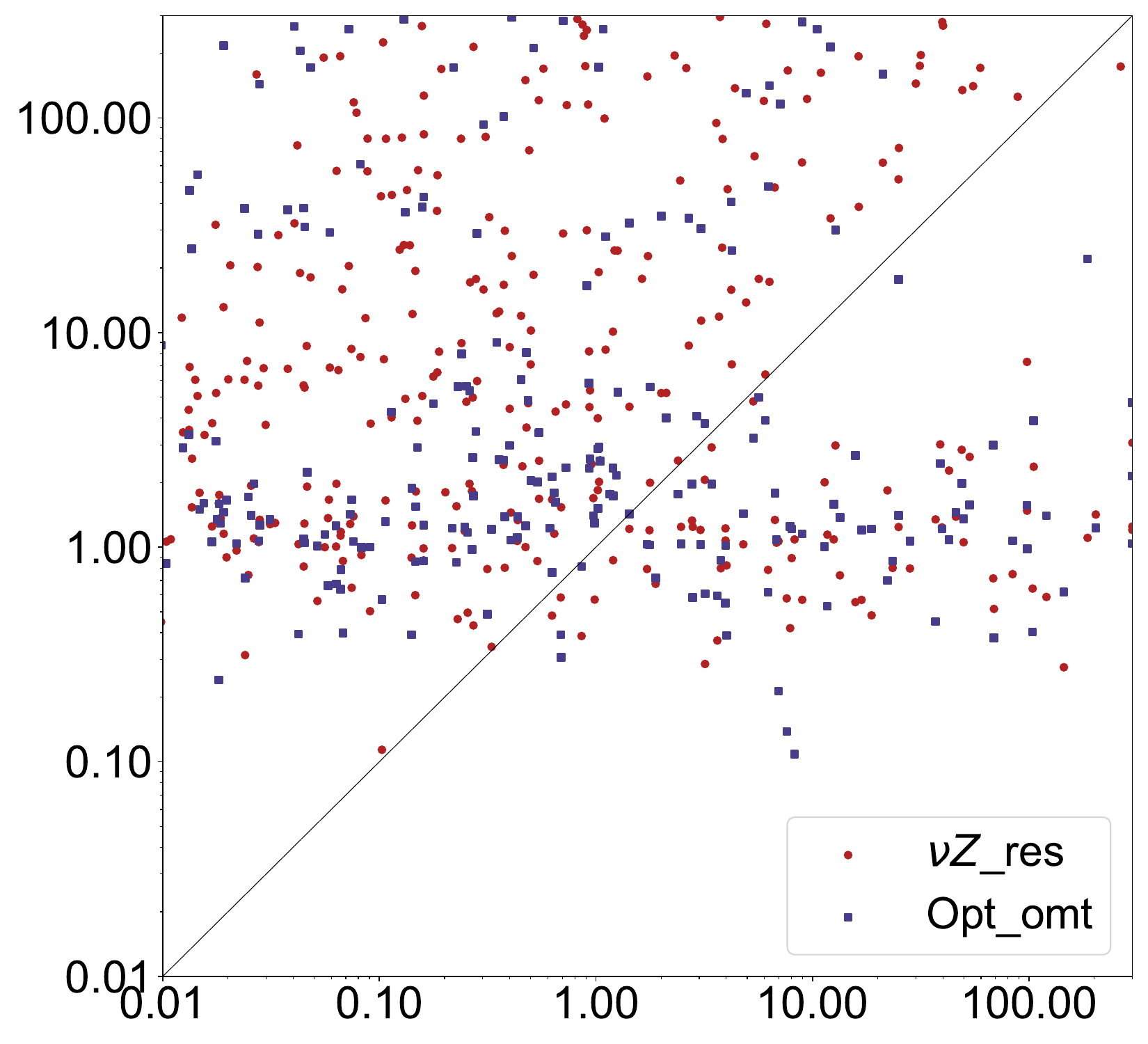}\vspace{-1.5mm}
        }
    \end{minipage}
    \begin{minipage}{0.24\linewidth}
        \subfigure[$SR=50\%$]{
            \includegraphics[width=\textwidth]{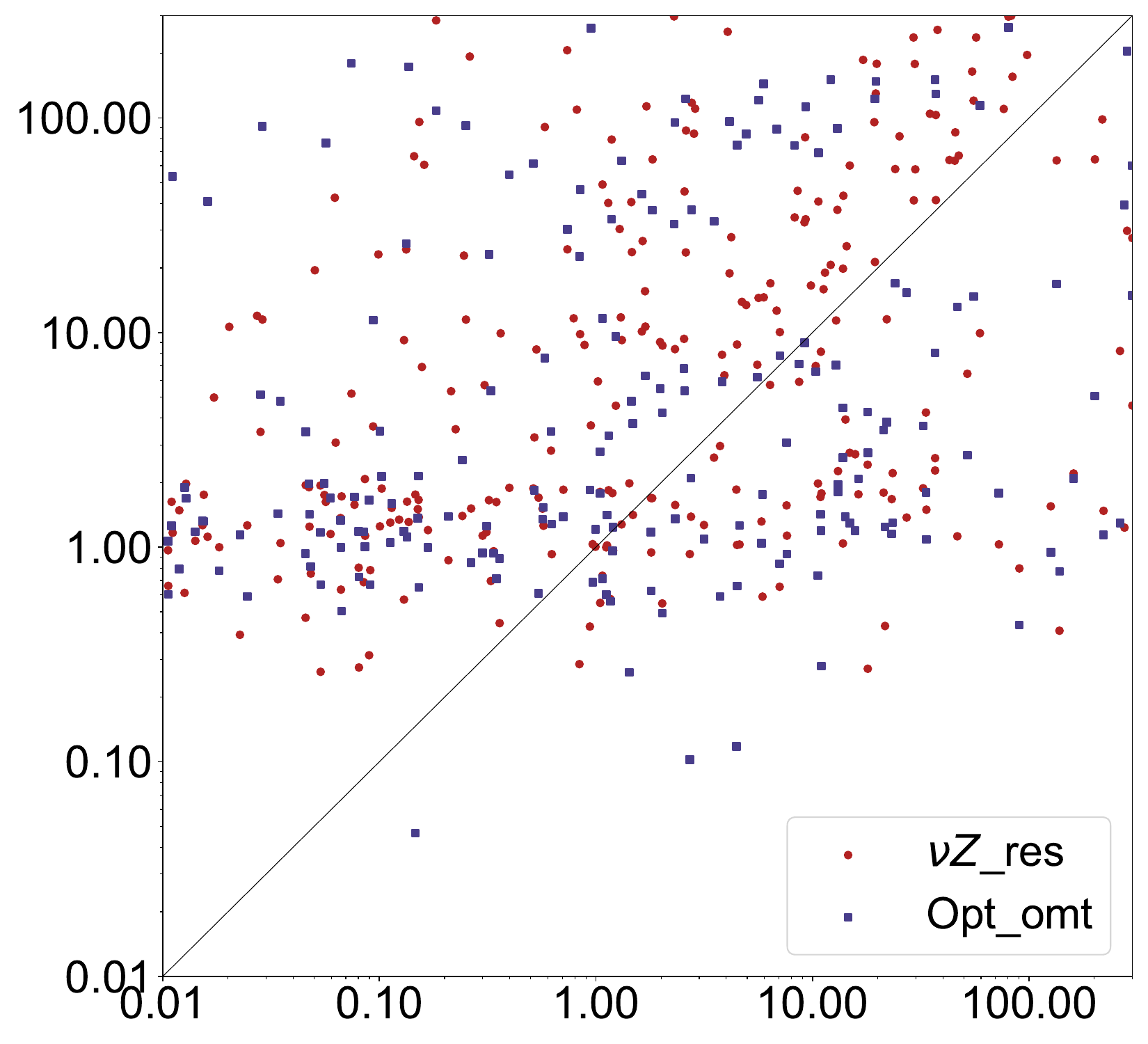}
        }
    \end{minipage}
    \begin{minipage}{0.24\linewidth}
        \subfigure[$SR=100\%$]{
            \includegraphics[width=\textwidth]{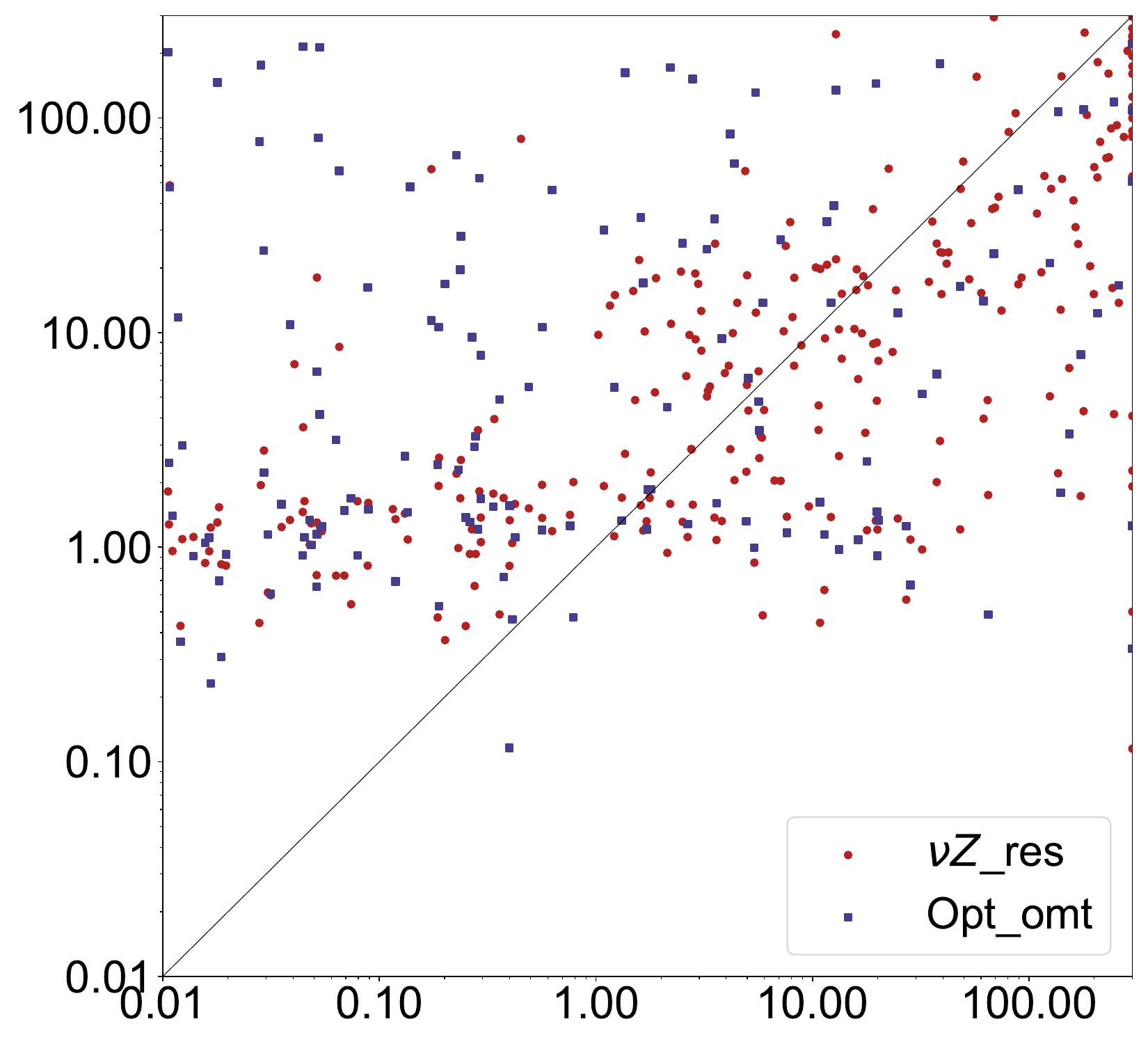}
        }
    \end{minipage}

    \caption{Run time comparison on MaxSMT-IDL}
    \label{compare_idl_fig}
\end{figure}



\section{Experiment of Evolution of Solution Quality}
\begin{table}[H]
\centering
\caption{The evolution of solution quality with different cutoff times in MaxSMT-LIA benchmark (LIA for short) and MaxSMT-IDL benchmark (IDL for short).}
\label{timelia_table}
\renewcommand\arraystretch{1.1}
\setlength{\tabcolsep}{1.0mm}\scalebox{1.0}{
{
\begin{tabular}{@{}llllllll@{}}
\toprule
\multirow{2}*{} & \multirow{2}*{{\it cutoff}} & \multirow{2}*{\#inst} 
&\ \ Opt\_omt&\ \ Opt\_res&
$\ \ \ \nu Z$\_res& $\ \ \nu Z$\_wmax 
&\ \ \  PairLS \\
&
&
&$\#win$($time$)
&$\#win$($time$)
&$\#win$($time$)
&$\#win$($time$)
&$\#win$($time$)\\ \midrule
\multirow{4}{*}{\it LIA} 
&50 &  5520 & 2869(10.6) & 3362(10.3) & 2906(8.2)& 1803(10.2) & \textbf{4531}(17.8)  \\ 
&100 & 5520 &3102(20.0)  &3712(10.5) &  2980(14.0) & 1802(18.9) & \textbf{4717}(25.3) \\ 
&200& 5520 & 3237(32.1) & 3914(10.2)  & 3130(19.6) & 1805(20.1) & \textbf{4763}(27.6) \\ 
&300& 5520 & 3267(34.3) & 3937(10.9)  &3168(26.1)  &1837(18.7)  & \textbf{4778}(30.7) \\ \midrule
\multirow{4}{*}{\it IDL} 
&50&  12888 &5239(22.2)  & 4589(20.9)  & 6373(24.7)& 4322(18.7) & \textbf{7175}(35.9)  \\ 
&100 & 12888 & 5220(42.1)  &4489(39.6) &  6369(47.2) & 4075(34.8) & \textbf{7014}(70.1) \\ 
&200& 12888 & 5091(138.7) &  4248(143.1)  &  6063(130.1)&  3881(134.3) &    \textbf{6968}(131.4) \\ 
&300& 12888 & 4412(210.5)  &  3577(190.2)  &   5673(205.8)&   3253(264.5) &   \textbf{6897}(153.3) \\ \bottomrule

\end{tabular}
}
}
\end{table}

\section{Experiment of Extended Neighborhood}
We conduct our experiments on MaxSMT-LIA. The results are in Table \ref{tuple_table}. \#Avergage Steps represents the number of iterations of the algorithm. 
\begin{table}[]
    \centering
    \caption{Extend the number of compensated variables on MaxSMT-LIA.
}
    \label{tuple_table}
    \renewcommand\arraystretch{1.0}
\setlength{\tabcolsep}{7mm}\scalebox{1.0}{
        \begin{tabular}{@{}llll@{}}
            \toprule

            & $v_{no\_pair}$  & PairLS & $v_{tuple}$ \\ \midrule
            \#win & 3686 & {\bf 5455} & 1029  \\
            \#Average Steps & {\bf 1952354} & 356752& 35850 \\ \bottomrule
        \end{tabular}
    }
\end{table}

\section{Experiment of SMT-LIB}

{\bf Competitors: }
We compare LS-LIA\_Pair with LS-LIA and  4 state-of-the-art complete SMT solvers according to  SMT-COMP 2023\footnote{https://smt-comp.github.io/2023}, namely
MathSAT5(version 5.6.6),
CVC5(version 0.0.4),
Yices2(version 2.6.2),
and Z3(version 4.8.14).
The binaries of all competitors are downloaded from their websites.

{\bf Benchmarks: } 
This benchmark consists of SMT(LIA) instances from SMT-LIB\footnote{https://clc-gitlab.cs.uiowa.edu:2443/SMT-LIB-benchmarks/QF\_LIA}. UNSAT instances are excluded, resulting in a benchmark consisting of 6670 unknown and satisfiable instances. 
Cutoff time is set as 1200 seconds.
\begin{table}[h]
\centering
\caption{Results on instances from SMTLIB.
}
\label{tbl:lia}
\renewcommand\arraystretch{1.0}
\setlength{\tabcolsep}{1.3mm}\scalebox{0.85}{
\begin{tabular}{@{}lllllllll@{}}
\toprule
Family                                                                      & Type                  & \#inst & MathSAT5     & CVC5         & Yices2       & Z3           & LS-LIA     & LS-LIA-Pair   \\ \midrule
\multirow{14}{*}{\begin{tabular}[c]{@{}l@{}}Without\\ Boolean\end{tabular}} & 20180326-Bromberger   & 631    & 123 & 206          & 273          & 99          & \textbf{50}         & \textbf{50}     \\
                                                                            & bofill-scheduling     & 407    & \textbf{0} & 5         & \textbf{0} & 2          & 16        & 10(-6)   \\
                                                                            & CAV\_2009\_benchmarks & 506    & \textbf{0} & 8         & 110          & \textbf{0} & \textbf{0} & \textbf{0}  \\
                                                                            & check                 & 0      & 0            & 0            & 0            & 0            & 0           & 0  \\
                                                                            & convert               & 280    & 7          & 75         & 94          & 96          & \textbf{1}  & \textbf{1} \\
                                                                            & dillig                & 230    & \textbf{0} & \textbf{0} & 30          & \textbf{0} & \textbf{0}  & \textbf{0}\\
                                                                            & miplib2003            & 16     & 6           & 7            & 5          & 8            & \textbf{3}  & \textbf{3} \\
                                                                            & pb2010                & 41     & 27           & 36            & 20           & 8  & 13           & \textbf{6(-7)}   \\
                                                                            & prime-cone            & 19     & 0           & 0           & 0           & 0           & 0           & 0 \\
                                                                            & RWS                   & 20     & 9          & 7         & 9          & \textbf{6}  & 8          & 8  \\
                                                                            & slacks                & 231    & 1         & \textbf{0} & 70          & 1         & \textbf{0} & \textbf{0}  \\
                                                                            & wisa                  & 3      & 0           & 0           & 0            & 0            & 0            & 0  \\
                                                                            & SMPT(2022)            & 4285   & 85      &      84        &   \textbf{65}           &    \textbf{65}          &  101           &  94 (-7)           \\
                                                                            & Total                 & 6670   & 228        & 428        & 676        & 285        & 192 & \textbf{172(-20)}\\ \midrule
\multirow{14}{*}{\begin{tabular}[c]{@{}l@{}}With \\ Boolean\end{tabular}}   & 2019-cmodelsdiff      & 144    & 50           & \textbf{49}  & \textbf{49}  & \textbf{49}  & 93        & 76 (-17)    \\
                                                                            & 2019-ezsmt            & 108    & \textbf{24}  & 29           & 27           & 27           & 54       & 51 (-3)  \\
                                                                            & 20210219-Dartagnan    & 47     & 25          & 25           & \textbf{24}  & \textbf{24}  & 45        & 45     \\
                                                                            & arctic-matrix         & 100    & 57           & 74           & 41           & 53           & \textbf{23}  & \textbf{23} \\
                                                                            & Averest               & 9      & \textbf{0}   & \textbf{0}   & \textbf{0}   & \textbf{0}   & 2          & 2   \\
                                                                            & calypto               & 24      & \textbf{0}          & \textbf{0}           & \textbf{0}           & \textbf{0}           & 21         & 21    \\
                                                                            & CIRC                  & 18     & \textbf{0}  & \textbf{0}  & \textbf{0}  & \textbf{0}  & 15          & 15    \\
                                                                            & fft    & 5      & 2           & 2            & 2            & 2            & 2         & 2   \\
                                                                            & mathsat               & 21     & \textbf{0}  & \textbf{0}  & \textbf{0}  & \textbf{0}  & 8        & 8     \\
                                                                            & nec-smt               & 1256     & 12       &  831         & \textbf{0}&  14         & 675     & 675       \\
                                                                            & RTCL                  & 2      & 0            & 0           & 0            & 0            & 0       & 0        \\
                                                                            & tropical-matrix       & 108    & 53          & 66          & 37           & 56           & \textbf{10} & \textbf{10}  \\
                                                                            & Total                 & 1842    & 223         & 766          & \textbf{180} & 225          & 930      & 910(-20)      \\ \bottomrule
\end{tabular}
}
\end{table}

\end{document}